\documentclass[11pt]{article}

\usepackage{amsmath}
\usepackage{amsthm,amsfonts,amssymb}
\usepackage{graphicx}
\usepackage[dvipsnames]{xcolor}
\usepackage{fullpage}
\usepackage[numbers,sort&compress]{natbib}
\usepackage{enumitem}
\usepackage{lscape}
\usepackage{hyperref}       % hyperlinks
\hypersetup{colorlinks=true, citecolor=Blue}
\usepackage{soul}
\usepackage{bm}
\usepackage{tikz}
\usetikzlibrary{arrows.meta, positioning, fit, backgrounds, shapes.geometric, calc, arrows.meta}

\usepackage{setspace}
\usepackage{caption}
\newtheorem{lemma}{Lemma}

\newtheorem{theorem}{Theorem}
\newtheorem{corollary}{Corollary}

\theoremstyle{definition}

\DeclareMathOperator{\tr}{Tr}
\DeclareMathOperator{\diag}{diag}

\newcommand{\dd}{\operatorname{d}\!}
\newcommand{\norm}[1]{\lVert{#1}\rVert}

\newcommand{\inv}{^{-1}}

\renewcommand{\v}{{\bf v}}

\newcommand{\x}{{\bf x}}
\newcommand{\y}{{\bf y}}

\newcommand{\A}{{\bf A}}

\newcommand{\F}{{\bf F}}
\newcommand{\G}{{\bf G}}

\newcommand{\I}{{\bf I}}

\newcommand{\M}{{\bf M}}

\newcommand{\Q}{{\bf Q}}
\newcommand{\R}{\mathbb{R}}

\renewcommand{\S}{{\bf S}}

\newcommand{\U}{{\bf U}}
\newcommand{\V}{{\bf V}}
\newcommand{\W}{{\bf W}}

\newcommand{\calD}{\mathcal{D}}
\newcommand{\calE}{\mathcal{E}}

\newcommand{\calN}{\mathcal{N}}
\newcommand{\calO}{\mathcal{O}}
\newcommand{\calS}{\mathcal{S}}

\newcommand{\calU}{\mathcal{U}}

\newcommand{\calZ}{\mathcal{Z}}

\title{Global stability of a Hebbian/anti-Hebbian network for principal subspace learning}
\author{David~Lipshutz\thanks{Department of Neuroscience, Baylor College of Medicine and Neuroengineering Initiative, Rice University.}\; and Robert J.~Lipshutz}
\date{\today}

\begin{document}

\maketitle

\begin{abstract}
Biological neural networks self-organize according to local synaptic modifications to produce stable computations. 
How modifications at the synaptic level give rise to such computations at the network level remains an open question.
\citet{pehlevan2015hebbian} proposed a model of a self-organizing neural network with Hebbian and anti-Hebbian synaptic updates that implements an algorithm for principal subspace analysis; however, global stability of the nonlinear synaptic dynamics has not been established.
Here, for the case that the feedforward and recurrent weights evolve at the same timescale, we prove global stability of the continuum limit of the synaptic dynamics and show that the dynamics evolve in two phases. 
In the first phase, the synaptic weights converge to an invariant manifold where the `neural filters' are orthonormal. 
In the second phase, the synaptic dynamics follow the gradient flow of a non-convex potential function whose minima correspond to neural filters that span the principal subspace of the input data.
\end{abstract}

\section{Introduction}

Biological neural networks self-organize according to local synaptic interactions that produce stable computations at the network-level. 
A challenge in theoretical neuroscience is to link these local interactions to stable network-level computations.

In a seminal work, \citet{oja1982simplified} proposed a computational model of a neuron as implementing an online algorithm for learning the top principal component of its input data using local, Hebbian synaptic updates.
Oja's algorithm thus establishes a link between local, Hebbian interactions at the synaptic level and principal components analysis (PCA) at the neuron level.
Moreover, the online (stochastic) algorithm is strikingly stable and data efficient---it is globally stable with a convergence rate that matches the information theoretic lower bound \citep{chou2020ode}---suggesting that neural networks with local synaptic updates may also be relevant in machine learning and neuromorphic computing applications.

In the multi-channel setting, much less is known about the stability of neural networks with local synaptic updates.
Following Oja's work \citep{oja1982simplified}, there were several extensions to principal subspace analysis (PSA) algorithms that can be implemented in neural networks with local Hebbian and anti-Hebbian synaptic updates \citep{foldiak1989adaptive,rubner1989self,rubner1990development,leen1991learning,pehlevan2015hebbian,pehlevan2018similarity}; however, these multi-channel networks include recurrent interactions that complicate their analyses and global stability of the networks' dynamics has not been established.
In another line of work, networks with \textit{non-local} synaptic updates for PSA have been proposed and analyzed \citep{oja1983analysis,sanger1989optimal,leen1990hebbian,oja1985stochastic,hornik1992convergence,oja1992principal,plumbley1995lyapunov}, but these networks do not explain how local interactions give rise to stable network level computations.

\begin{figure}
    \centering
    \includegraphics{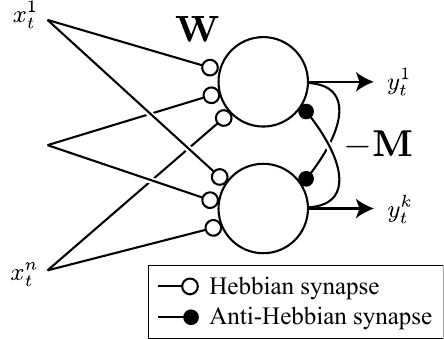}
    \caption{Hebbian/anti-Hebbian network for PSA. Single layer network with $k$ neurons that receives $n$ inputs. Feedforward Hebbian synapses $\W$ connect the $n$ inputs to the $k$ neurons and recurrent anti-Hebbian synapses $-\M$ connect the $k$ neurons.}
    \label{fig:network}
\end{figure}

The focus of this work is to analyze the global synaptic dynamics of a multi-channel network for principal subspace analysis with Hebbian feedforward synapses and anti-Hebbian recurrent synapses introduced by \citet{pehlevan2015hebbian,pehlevan2018similarity} and depicted in Figure \ref{fig:network}. Specifically, we show that the ordinary differential equation (ODE) associated with the synaptic dynamics is globally stable in the sense that for almost every initialization, the synaptic weights converge to an optimal configuration associated with PSA. We further show that the dynamics evolve in two phases. In the first phase, the synaptic weights converge to an invariant manifold characterized by orthonormality of the neural filters. In the second phase, the synaptic dynamics follow the gradient flow of non-convex potential function.

\section{Hebbian/anti-Hebbian network model}

In this section, we first review the Hebbian/anti-Hebbian neural network model for PSA proposed by \citet{pehlevan2015hebbian}.
A detailed derivation of the network from a so-called ``similarity matching'' objective can be found in \citep{pehlevan2015hebbian,pehlevan2018similarity}.
We then introduce the ODE associated with the synaptic dynamics, which will be the focus of our analysis in this work.

Consider a network with $k$ primary neurons that receives $n>k$ inputs, as illustrated in Figure \ref{fig:network}. 
Feedforward synapses $\W$ connect the inputs to primary neurons, and recurrent synapses $-\M$ connect the primary neurons. 
At each timestep $t=1,2,\dots$, the network receives an input vector $\x_t=(x_t^1,\dots,x_t^n)$.
The network operates on two timescales: a fast neural dynamics timescale and a slow synaptic update timescale.

In the first phase, the neural activities of the $k$ output neurons, which are denoted by the $k$-dimensional vector $\y_t=(y_t^1,\dots,y_t^k)$, evolve according to fast linear neural dynamics
\begin{align*}
    \dot\y_t(\gamma)=\W_t\x_t-\M_t\y_t(\gamma),
\end{align*}
which converge to $\y_t=\F_t\x_t$, where $\W_t$ and $\M_t$ respectively denote the states of the feedforward and recurrent synaptic weights at time $t$, and the row vectors of the $k\times n$ matrix $\F_t:=\M_t^{-1}\W_t$ are referred to as the `neural filters'. After the neural activities converge, the synaptic weights are updated according to the slow synaptic plasticity rules
\begin{align}
    \W_{t+1}&=\W_t+2\eta(\y_t\x_t^\top-\W_t) \label{eq:DW}\\ 
    \M_{t+1}&=\M_t+\frac\eta\tau(\y_t\y_t^\top-\M_t)\label{eq:DM},
\end{align}
where $\eta>0$ is the learning rate for the feedforward synapses $\W$ and $\tau$ denotes the ratio between the learning rates for the feedforward synapses $\W$ and the lateral synapses $-\M$. The plasticity rule for the feedforward synapses $\W$ includes a term proportional to the product of the pre- and postsynaptic activities, $\y_t\x_t^\top$, so it is referred to as `Hebbian'. The plasticity rule for the lateral synapses $-\M$ includes a term inversely proportional to the product of the pre- and postsynaptic activities, $\y_t\y_t^\top$, so it is referred to as `anti-Hebbian'.

To analyze the stability of their algorithm, \citet{pehlevan2018similarity} considered the continuum limit of the updates. Formally, when the input data $\{\x_t\}$ are independent and identically distributed samples with zero mean and fixed $n\times n$ covariance matrix $\A$ and the step size $\eta>0$ is infinitesimally small, the synaptic dynamics can be approximated by the ODE
\begin{align}
    \label{eq:dW}
    \frac12\frac{\dd \W(t)}{\dd t}&=\M(t)^{-1}\W(t) \A -\W(t)\\
    \label{eq:dM}
    \tau\frac{\dd \M(t)}{\dd t}&=\M(t)^{-1}\W(t) \A \W(t)^\top \M(t)^{-1}-\M(t).
\end{align}
The relationship between the online algorithm and the ODE can be made precise for certain time-dependent learning rates under appropriate regularity conditions (e.g., the spectrum of $\M$ is uniformly bounded away from zero) \citep{borkar1997stochastic,kushner2012stochastic}. 

\citet{pehlevan2018similarity} proved that every equilibrium point of the ODE corresponds to an eigen-subspace of $\A$, and, when $\tau\le\frac12$, all linearly stable equilibrium points correspond to the \textit{principal} (eigen-)subspace of $\A$. While their theoretical analysis is informative about the synaptic dynamics near the equilibrium points, it is not informative about the synaptic dynamics away from the equilibrium points, which corresponds to most random initializations. \citet{pehlevan2015hebbian,pehlevan2018similarity} provide empirical evidence that both the online algorithm and ODE are globally stable, and they benchmark both against comparable algorithms; however, a theoretical analysis of the global dynamics remains an open problem. The focus of this work is to prove global stability of the ODE.

The ODE \eqref{eq:dW}--\eqref{eq:dM} is naturally viewed as the gradient descent-ascent flow (with timescale separation $\tau$) for solving the nonconvex-concave minimax problem
\begin{align}\label{eq:minimax}
    \min_\W\max_\M f(\W,\M),&&f(\W,\M):=\tr\left(-\M^{-1}\W\A\W^\top+\W\W^\top-\frac12\M^2\right).
\end{align}
where the minimization is over the set of $k\times n$ matrices $\R^{k\times n}$ and the maximization is over the set of $k\times k$ positive definite matrices $\calS_{++}^k$. Analogously, the discrete algorithm \eqref{eq:DW}--\eqref{eq:DM} is naturally interpreted as a stochastic gradient descent-ascent algorithm for solving the minimax problem. In general, proving convergence of gradient descent-ascent algorithms for nonconvex-concave minimax problems is challenging and existing results \citep{borkar1997stochastic,lin2020gradient,centorrino2025similarity} rely on a separation of time-scales (i.e., letting $\tau\to0$).\footnote{During the preparation of this manuscript, we became aware of the preprint \citep{centorrino2025similarity} that also analyzes the stability of the ODE \eqref{eq:dW}--\eqref{eq:dM}. However, the analysis considers the regime $\tau\to0$ whereas we treat the case $\tau=\frac12$.} However, numerical experiments indicate that the optimal convergence rate occurs around $\tau=\frac12$, see \cite[Figure 4]{pehlevan2018similarity}, suggesting that a separation of time-scales is unnecessary for proving global convergence and that $\tau=\frac12$ may be a parameter of particular interest.

\section{Global stability of the synaptic dynamics}

For the case $\tau=\frac12$, we prove global convergence of the ODE \eqref{eq:dW}--\eqref{eq:dM} to the desired principal subspace. We assume $\A$ is a positive definite $n\times n$ matrix with eigenvalues $\lambda_1\ge\cdots\ge\lambda_n>0$ and $k<n$ is fixed. The following theorem is our main result.

\begin{theorem}\label{thm:aeconvergence}
    Let $\A\in\calS_{++}^n$.
    For every $\tau>0$ and $(\W_0,\M_0)\in\calD:=\R^{k\times n}\times\calS_{++}^k$, there exists a unique solution $(\W(t),\M(t))$ of the ODE \eqref{eq:dW}--\eqref{eq:dM} with initial condition $(\W_0,\M_0)$ for all $t\ge0$. Moreover, suppose $\lambda_k>\lambda_{k+1}$ and $\tau=\frac12$. Then there is a set $\calZ$ with Lebesgue measure zero such that if $(\W_0,\M_0)\in\calD\setminus\calZ$, then the solution $(\W(t),\M(t))$ converges, as $t\to\infty$, to the set of equilibrium points $(\W_\ast,\M_\ast)$ of the ODE such that the neural filters (i.e., the row vectors of $\F_\ast:=\M_\ast^{-1}\W_\ast$) are orthonormal and span the principal subspace of $\A$.
\end{theorem}

The next two sections and appendix \ref{apdx:exist} are devoted to the proof of Theorem \ref{thm:aeconvergence}. 
In appendix \ref{apdx:exist} we prove existence and uniqueness of solutions, which ensures that solutions that are initialized in $\calD$ remain in $\calD$ for all $t\ge0$.
In section \ref{sec:local}, we review results from \cite{pehlevan2018similarity} on local linear stability of equilibrium points, denoted $\calE$. Global convergence is proved in section \ref{sec:stable}. We show that the synaptic weights evolve in two phases. In the first phase, for any initialization outside of a null set $\calN$, the synaptic weights first converge to the invariant manifold $\calO$ corresponding to orthonormal neural filters; that is, the set of $(\W,\M)\in\calD$ such that $\F:=\M^{-1}\W$ has orthonormal row vectors (section \ref{sec:Oconvergence} and Figure \ref{fig:vectorplot}). This convergence is captured by a convex Lyapunov function $L(\W,\M)$. Then, in the second phase, starting on (or near) the invariant manifold $\calO$, the synaptic dynamics are approximated by the gradient flow of a non-convex potential function $V(\W)$ (section \ref{sec:gradflow} and Figure \ref{fig:gradflow}). As a result, for almost any initialization, the synaptic weights converge to an equilibrium point such that the neural filters correspond to the desired principal subspace projection. 
A dependency diagram for our main result is shown in Figure \ref{fig:outline}.

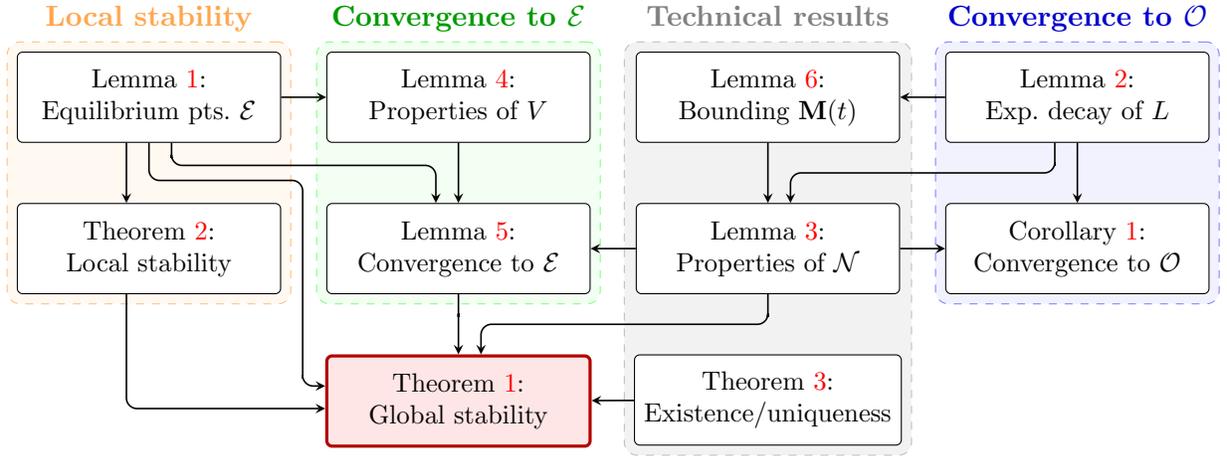
\begin{figure}
    \centering
    \begin{tikzpicture}[
    node distance=0.8cm and 0.6cm,
    block/.style={rectangle, draw, fill=white, minimum width=3.5cm, minimum height=1.2cm, align=center, rounded corners=2pt, font=\small},
    arrow/.style={-stealth, semithick},
    cat_local/.style={draw=orange!50, fill=orange!5, dashed, rounded corners=5pt},
    cat_equil/.style={draw=green!50, fill=green!5, dashed, rounded corners=5pt},
    cat_inv/.style={draw=blue!50, fill=blue!5, dashed, rounded corners=5pt},
    cat_tech/.style={draw=gray!50, fill=gray!10, dashed, rounded corners=5pt}
]

% --- Nodes ---
% Column 1
\node[block] (L1) {Lemma \ref{lem:criticalpts}: \\ Equilibrium pts.\ $\calE$};
\node[block, below=of L1] (T2) {Theorem \ref{thm:localstability}: \\ Local stability};

% Column 2
\node[block, right=of L1] (L4) {Lemma \ref{lem:V}: \\ Properties of $V$};
\node[block, below=of L4] (L5) {Lemma \ref{lem:globalconv}: \\ Convergence to $\calE$};
\node[block, below=of L5, draw=red!70!black, very thick, fill=red!10, rounded corners=2pt] (T1) {Theorem \ref{thm:aeconvergence}: \\ Global stability};

% Column 3 (Technical Results)
\node[block, right=of L4] (L6) {Lemma \ref{lem:detM}: \\ Bounding $\M(t)$};
\node[block, below=of L6] (L3) {Lemma \ref{lem:setN}: \\ Properties of $\calN$};
\node[block, below=of L3] (T3) {Theorem \ref{thm:exist}: \\ Existence/uniqueness};

% Column 4 (Invariant Manifold)
\node[block, right=of L6] (L2) {Lemma \ref{lem:WMconv}: \\ Exp.\ decay of $L$};
\node[block, below=of L2] (C1) {Corollary \ref{cor:MWinv}: \\ Convergence to $\calO$};

% --- Background Layer (Boxes and Arrows) ---
\begin{scope}[on background layer]
    % Category Boxes
    \node[cat_local, fit=(L1) (T2), label={[text=orange!70, font=\bfseries]above:Local stability}] {};
    \node[cat_equil, fit=(L4) (L5), label={[text=green!60!black, font=\bfseries]above:Convergence to $\calE$}] {};
    \node[cat_tech, fit=(L6) (L3) (T3), label={[text=gray, font=\bfseries, yshift=2pt]above:Technical results}] {};
    \node[cat_inv, fit=(L2) (C1), label={[text=blue!80!black, font=\bfseries]above:Convergence to $\calO$}] {};

    % Straight Arrows behind nodes
    \begin{scope}[arrow]
        % L1 connections
        \draw (L1) -- (L4);
        \draw[arrow] ([xshift=-3mm]L1.south) -- ([xshift=-3mm]T2.north);
        \draw[arrow, rounded corners] ([xshift=3mm]L1.south) |- ([xshift=3mm, yshift=5mm]L5.north -| L1.south) -| ([xshift=-3mm]L5.north);
        \draw[arrow, rounded corners=4pt] (L1.south) |- ([xshift=3mm, yshift=3mm]L5.north -| L1.south) -| ([xshift=-3mm, yshift=2mm]T1.west) |- ([yshift=2mm]T1.west);

        % L4 and T2 connections
        \draw (L4) -- (L5);
        \draw[arrow, rounded corners=4pt] ([xshift=-3mm]T2.south) |- ([yshift=-1mm]T1.west);

        % L2 connections
        \draw (L2) -- (L6);
        \draw[arrow, rounded corners] ([xshift=-3mm]L2.south) |- ([xshift=-3mm, yshift=4mm]L3.north -| L2.south) -| ([xshift=3mm]L3.north);
        \draw (L2) -- (C1);

        % L5 connections
        \draw[arrow, rounded corners=4pt] (L5.south) |- ([yshift=4mm]T1.north -| L5.south) -| (T1.north);

        % L3 connections
        \draw (L3) -- (L5);
        \draw (L3) -- (C1);
        \draw[arrow, rounded corners] (L3.south) |- ([yshift=4mm]T1.north -| L3.south) -| ([xshift=3mm]T1.north);

        % L6 and T3 connections
        \draw (L6) -- (L3);
        \draw (T3) -- (T1);
    \end{scope}
\end{scope}

\end{tikzpicture}
    \caption{
    Dependency diagram for our results. 
    Local stability of equilibrium points (orange region) is established in section \ref{sec:local}.
    Convergence of solutions to the invariant manifold $\calO$ (blue region) is shown in section \ref{sec:Oconvergence}.
    Convergence starting in or near the invariant manifold to the equilibrium points $\calE$ (green region) is shown in section \ref{sec:gradflow}.
    Finally, global stability of the ODE (red box) is shown in section \ref{sec:proof}.
    Technical results (gray region) are proved in appendices \ref{apdx:exist} and \ref{apdx:setN}.}
    \label{fig:outline}
\end{figure}

In section \ref{sec:stochastic}, we provide empirical evidence that the synaptic weights also evolve in two phases for the discrete online algorithm \eqref{eq:DW}--\eqref{eq:DM}. 
In section \ref{sec:tau}, we conjecture that Theorem \ref{thm:aeconvergence} can be generalized to hold for all $0<\tau\le\frac12$.

\section{Local stability of equilibrium points}\label{sec:local}

Next, we characterize the equilibrium points of the ODE \eqref{eq:dW}--\eqref{eq:dM} and recall results by \citet{pehlevan2018similarity} on their linear stability. To this end, let 
\begin{align*}
    \calE:=\{(\W,\M)\in\calD:\G(\W,\M)={\bf 0}\}    
\end{align*}
denote the set of equilibrium points, where $\G:\calD\mapsto\calD$ is the vector field defined by
\begin{align*}
    \G(\W,\M):=\left(2\M^{-1}\W\A-2\W,\frac1\tau(\M^{-1}\W\A\W^\top\M^{-1}-\M)\right).
\end{align*}

\subsection{Characterization of equilibrium points}

The following lemma, whose proof is given in appendix \ref{apdx:criticalpts}, characterizes the equilibrium points in terms of their singular value decompositions (SVDs).

\begin{lemma}\label{lem:criticalpts}
    Assume $\tau>0$ and $\A\in\calS_{++}^n$.  A pair $(\W_\ast,\M_\ast)\in\calE$ if and only if $\W_\ast=\U\S\V^\top$ and $\M_\ast=\U\S\U^\top$, where $\U$ is any $k\times k$ orthogonal matrix, $\V$ is a $d\times k$ matrix whose column vectors are orthonormal eigenvectors of the covariance matrix $\A$, and $\S$ is a $k\times k$ diagonal matrix whose diagonal entries are the eigenvalues of $\A$ corresponding to the column vectors of $\V$.
\end{lemma}

As a consequence of Lemma \ref{lem:criticalpts}, if $(\W_\ast,\M_\ast)\in\calE$, then $\W_\ast$ is full rank (since we assume the eigenvalues of $\A$ are positive) and $\F_\ast:=\M_\ast^{-1}\W_\ast=\U\V^\top$, where $\U$ is a $k\times k$ orthogonal matrix and the column vectors of $\V$ are orthonormal eigenvectors of $\A$. In other words, the row vectors of $\F_\ast$ (i.e., the neural filters) are orthonormal and span an eigen-subspace of $\A$.
In addition, if $(\W_\ast,\M_\ast)\in\calE$, then $(\Q\W_\ast,\Q\M_\ast\Q^\top)\in\calE$ for every $k\times k$ orthogonal matrix $\Q$. In particular, each equilibrium point is an element of a $\frac{k(k-1)}{2}$ dimensional manifold of equilibrium points corresponding to an eigen-subspace of $\A$.

\subsection{Local linear stability analysis}    
    
For our purposes, we say an equilibrium point $(\W_\ast,\M_\ast)\in\calE$ is `linearly stable' if all of the eigenvalues of the Jacobian of $\G$ evaluated at $(\W_\ast,\M_\ast)$ have nonpositive real part and it is `linearly unstable' if it is not linearly stable; that is, at least one of the eigenvalues of the Jacobian of $\G$ evaluated at $(\W_\ast,\M_\ast)$ has positive real part. Let
\begin{align*}
    \calE_0:=\{(\W,\M)\in\calE:\text{the rows of $\W$ span the principal subspace of $\A$}\}.    
\end{align*}
The next result follows immediately from \citep[Theorem 1]{pehlevan2018similarity}.

\begin{theorem}
\label{thm:localstability}
    Suppose $0<\tau\le\frac12$ and $\A\in\calS_{++}^n$. An equilibrium point $(\W_\ast,\M_\ast)\in\calE$ is linearly stable if and only if $(\W_\ast,\M_\ast)\in\calE_0$.
\end{theorem}

\section{Global convergence analysis}
\label{sec:stable}

We now prove our main results on global stability of the ODE \eqref{eq:dW}--\eqref{eq:dM} for $\tau=\frac12$, which is fixed throughout this section. Figure \ref{fig:vectorplot} shows the vector field $\G(W,M)$ in the scalar case $k=n=1$. 
In section \ref{sec:Oconvergence}, we show that the neural filters are asymptotically orthonormal---in Figure \ref{fig:vectorplot}, this corresponds to the initial convergence of trajectories to the blue line. Then, in section \ref{sec:gradflow}, we show that starting from (or near) an orthonormal initialization, the ODE \eqref{eq:dW} governing the dynamics of the feedforward weights $\W$ can be approximated as gradient flow of a potential function whose minima correspond to the principal subspace---in Figure \ref{fig:vectorplot}, this corresponds to the convergence of trajectories with initial conditions on (or near) the blue line to the red dots. Finally, in section \ref{sec:proof}, we combine these results to prove Theorem \ref{thm:aeconvergence}.

\begin{figure}
    \centering
    \includegraphics[width=.8\linewidth]{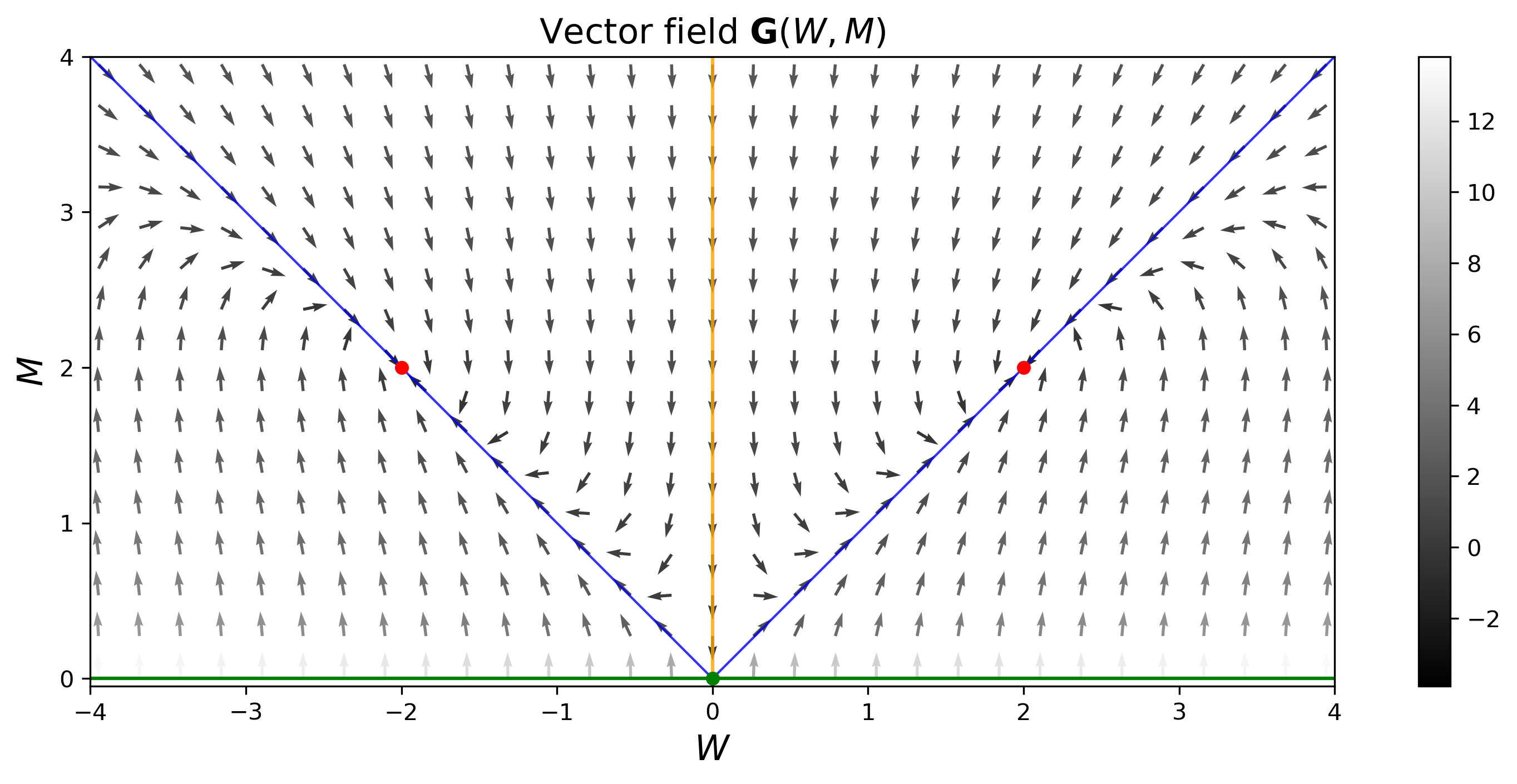}
    \caption{Plot of the vector field $\G(W,M)$ in the case $k=n=1$ and $\lambda_1=2$. The grayscale indicates the \textit{logarithm} of the vector magnitude. The blue lines denote the set $\calO$, the orange vertical line denotes the set $\calN$, the 2 red dots denote the set $\calE_0$ (which is equal to $\calE$ in this case), and the green line indicates that the line $M=0$ does not belong to $\calD=\R\times(0,\infty)$.}
    \label{fig:vectorplot}
\end{figure}

\subsection{Asymptotic orthonormality of the neural filters}
\label{sec:Oconvergence}

In this section, we show that for almost every initialization, the solution $(\W(t),\M(t))$ to the ODE \eqref{eq:dW}--\eqref{eq:dM} converges to the following invariant subset of matrices in $\calD$ that correspond to orthonormal neural filters $\F=\M^{-1}\W$:
\begin{align}\label{eq:Odef}
    \calO:=\left\{(\W,\M)\in\calD:\M^{-1}\W\W^\top\M^{-1}=\I_k\right\}.
\end{align}
To show this convergence, we define the following convex Lyapunov function on $\calD$:
\begin{align}
    L(\W,\M):=\norm{\W\W^\top-\M^2}^2=\tr\left[(\W\W^\top-\M^2)^2\right]. \label{eq:L}
\end{align}
Note that $L(\W,\M)$ is nonnegative everywhere and equal to zero if and only $(\W,\M)\in\calO$. 
% We show that for any initialization, $L(\W(t),\M(t))$ converges to zero as $t\to\infty$. 
% and, consequently, the neural filters are asymptotically orthonormal provided $(\W_0,\M_0)\not\in\calN$.

\begin{lemma}\label{lem:WMconv}
Suppose $\A\in\calS_{++}^n$ and $(\W(t),\M(t))$ is a solution of the ODE \eqref{eq:dW}--\eqref{eq:dM} with $\tau=\frac12$ and initial condition $(\W_0,\M_0)\in\calD$. Then
\begin{align}\label{eq:Ldecay}
    L(\W(t),\M(t))=L(\W_0,\M_0)e^{-8t},\qquad t\ge0.
\end{align}
\end{lemma}

\begin{proof}
First, note that the ODE \eqref{eq:dW}--\eqref{eq:dM} implies
\begin{align*}
    \frac{\dd\W(t)}{\dd t}\W(t)^\top-\frac{\dd\M(t)}{\dd t}\M(t)&=2\M(t)^{-1}\W(t) \A\W(t)^\top -2\W(t)\W(t)^\top\\
    &\qquad-\left[2\M(t)^{-1}\W(t)\A\W(t)^\top\M(t)^{-1}\M(t)-2\M(t)^2\right]\\
    &=-2\left[\W(t)\W(t)^\top-\M(t)^2\right].
\end{align*}
Therefore, by the chain rule, the product rule, the cyclic property of the trace rule, and the previous display,
\begin{align*}
    \frac{\dd}{\dd t}L(\W(t),\M(t))&=4\tr\left[\left(\W(t)\W(t)^\top-\M(t)^2\right)\left(\frac{\dd \W(t)}{\dd t}\W(t)^\top-\frac{\dd \M(t)}{\dd t}\M(t)\right)\right]\\
    &=-8L(\W(t),\M(t)).
\end{align*}
Solving the differential equation by separation of variables yields equation \eqref{eq:Ldecay}. 
\end{proof}

To prove that the neural filters are asymptotically orthonormal, we need a technical lemma that states conditions under which the solution $(\W(t),\M(t))$ does not converge to zero.
To this end, define the set
\begin{equation}\label{eq:calN}
    \calN=\{(\W,\M)\in\calD: \exists\,\v\in\R^k,\,\lambda>0, \text{ such that}\;\W^\top\v={\bf 0}\text{ and } \M\v=\lambda\v\}.
\end{equation}
Then $\calN$ is the set of pairs $(\W,\M)\in\calD$ such that $\W\W^\top$ is singular \textit{and} there is an eigenvector of $\M$ in the null space of $\W\W^\top$. 
Note that if $\W$ is full rank and $\M$ is positive definite, then $(\W,\M)\not\in\calN$.
The set $\calN$ corresponds to the orange vertical line in Figure \ref{fig:vectorplot}. The following technical lemma establishes that if for any initialization not in $\calN$, the solution to the ODE \eqref{eq:dW}--\eqref{eq:dM} remains bounded away from zero and infinity.
The proof is provided in appendix \ref{apdx:setN}.

\begin{lemma}\label{lem:setN}
The set $\calN$ has Lebesgue measure zero. Suppose $\A\in\calS_{++}^n$ and $(\W(t),\M(t))$ is the solution of the ODE \eqref{eq:dW}--\eqref{eq:dM} with $\tau=\frac12$ and starting at $(\W_0,\M_0)\in\calD$. If $(\W_0,\M_0)\not\in\calN$, then $(\W(t),\M(t))\not\in\calN$ for all $t\ge0$ and
\begin{align*}
    \limsup_{t\to\infty}\left\{\|\M(t)^{-1}\|+\|\W(t)\|\right\}<\infty.
\end{align*}
On the other hand, if $(\W_0,\M_0)\in\calN$, then $(\W(t),\M(t))\in\calN$ for all $t\ge0$ and 
\begin{align*}
    \lim_{t\to\infty}\det(\M(t))=0.
\end{align*}
\end{lemma}

The following corollary of Lemmas \ref{lem:WMconv} and \ref{lem:setN} states that almost every solution of the ODE converges exponentially to the invariant manifold $\calO$.

\begin{corollary}\label{cor:MWinv}
    Suppose $\A\in\calS_{++}^n$ and $(\W(t),\M(t))$ is a solution of the ODE \eqref{eq:dW}--\eqref{eq:dM} with  $\tau=\frac12$ and initial condition $(\W_0,\M_0)\in\calD\setminus\calN$. 
    Then $(\W(t),\M(t))$ converges to $\calO$ as $t\to\infty$.
\end{corollary}

\begin{proof}
    By Lemma \ref{lem:setN}, $K:=\sup\{\|\M(t)^{-1}\|:t\ge0\}<\infty$. Thus,
    \begin{align*}
        \|\M(t)^{-1}\W(t)\W(t)^\top\M(t)^{-1}-\I_k\|^2&\le\|\M(t)^{-2}\|^2\|\W(t)\W(t)^\top-\M(t)^2\|^2\\
        &\le K^4e^{-8t}\|\W_0\W_0^\top-\M_0^2\|^2,
    \end{align*}
    where the first inequality is due to the Cauchy-Schwarz inequality, and the second inequality follows from Lemma \ref{lem:WMconv}. 
\end{proof}

\subsection{Convergence to equilibrium points}\label{sec:gradflow}

Having shown that $(\W(t),\M(t))$ converges to $\calO$ as $t\to\infty$, we now analyze the dynamics when $(\W(t),\M(t))$ is near the set $\calO$. To begin, consider the case $(\W_0,\M_0)\in\calO$ so that $(\W(t),\M(t))\in\calO$ for all $t\ge0$. Then we can rewrite the right-hand side of the ODE \eqref{eq:dW} as a function of $\W(t)$ only:

\begin{align}
\label{eq:dW1}
    \frac{\dd \W(t)}{\dd t}=2(\W(t)\W(t)^\top)^{-\frac12}\W(t)\A-2\W(t)=-2\nabla V(\W(t)),
\end{align}
where $V:\R^{k\times n}\mapsto\R$ is the nonconvex potential function
\begin{align}
\label{eq:V}
    V(\W):=f\left(\W,(\W\W^\top)^\frac12\right)=\tr\left[-(\W\W^\top)^{-\frac12}\W\A\W^\top+\frac12\W\W^\top\right],
\end{align}
where $f(\W,\M)$ is defined in \eqref{eq:minimax}.
Therefore, we can interpret the $\W(t)$ dynamics on $\calO$ as the gradient flow of the potential function $V$. (Note that $V$ is only differentiable on the subset of full-rank matrices $\W$ in $\R^{k\times n}$.) In Figure \ref{fig:gradflow}, we plot the vector field $-\nabla V(\W)$ in the case $d=2$ and $k=1$ to illustrate the dynamics of $\W(t)$ on the set $\calO$.

\begin{figure}
    \centering
    \includegraphics[width=.8\linewidth]{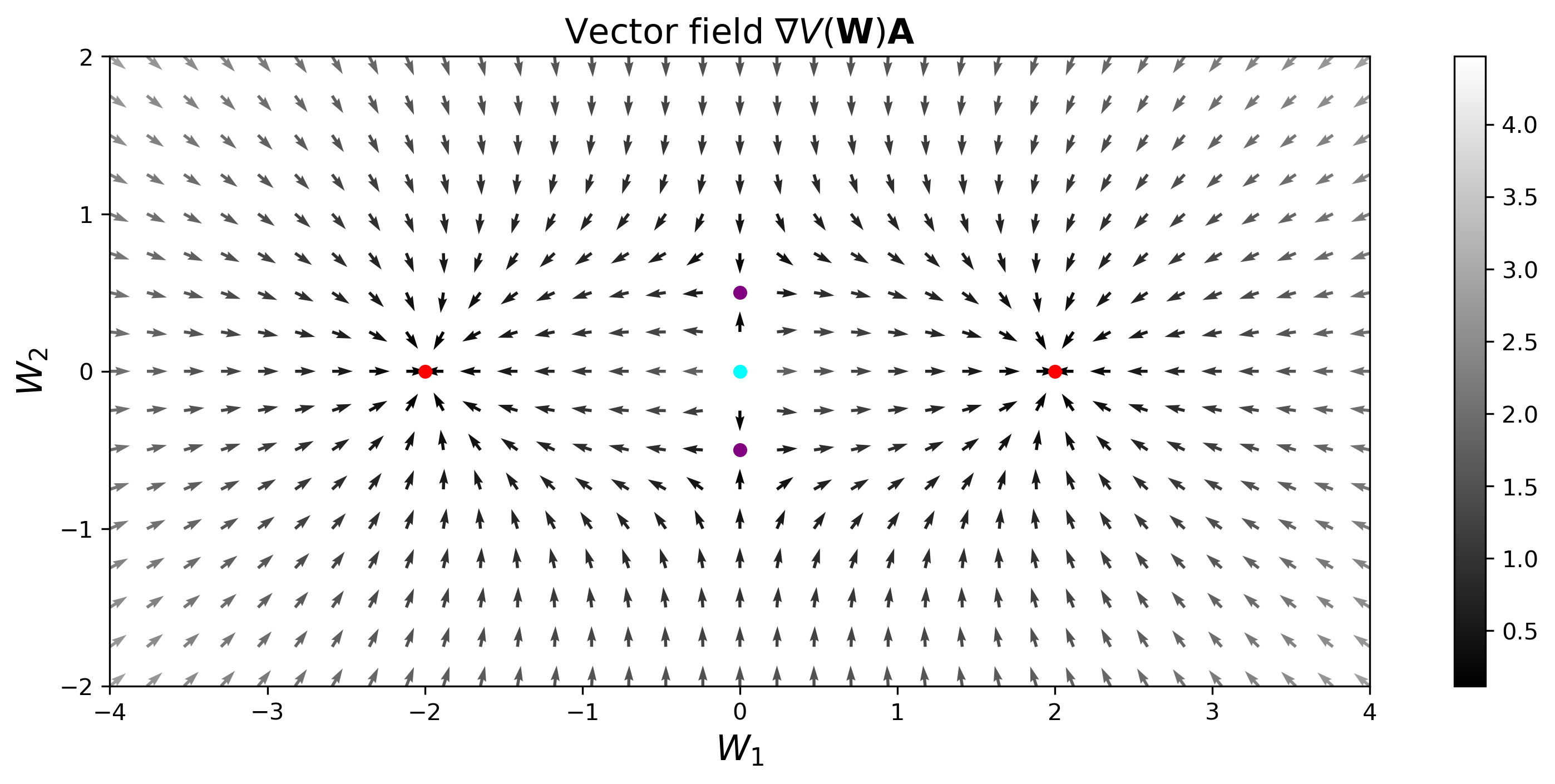}
    \caption{Plot of the vector field $-\nabla V(\W)$ in the case $n=2$, $k=1$ and $\A=\diag(2,\frac12)$. The grayscale indicates the vector magnitude. The red dots denote the global minima of $V$, the purple dots denote the saddle points of $V$, and the cyan dot at the origin denotes the set $\{\W:\det(\W\W^\top)=0\}$.}
    \label{fig:gradflow}
\end{figure}

\begin{lemma}\label{lem:V}
    Suppose $\A\in\calS_{++}^n$. The function $V$ is bounded below.  Furthermore, $\nabla V(\W)$ exists and satisfies $\nabla V(\W)=0$ if and only if $\W$ is full rank and $(\W,(\W\W)^{\frac12})\in\calE$.
\end{lemma}

\begin{proof}
    Let $\W\in\R^{k\times n}$ and let $\W=\U\S\V^\top$ denote its SVD. Then
    \begin{align*}
        V(\W) & \geq \tr\left[-\sigma_{\max}(\A)(\W\W^\top)^{\frac12}+\frac12\W\W^\top\right]\\
        &\ge-\sigma_{\max}(\A)\tr(\S)+\frac12\tr(\S^2)\\
        &\ge -\frac k2\sigma_{\max}^2(\A),
    \end{align*}
    where $\sigma_{\max}(\A)$ denotes the spectral norm of $\A$. Since this holds for all $\W$, the function $V$ is bounded below. Next, suppose $\W$ is full rank and $\nabla V(\W)=0$. Then $2(\W\W^\top)^{-\frac12}\W\A\W^\top=(\W\W^\top)^{\frac12}$ and so $\V^\top\A\V=\S$. Since $\S$ is diagonal, it follows that the column vectors of $\V$ are eigenvectors of $\A$ and the diagonal entries of $\S$ are the corresponding eigenvalues of $\A$. Thus, by Lemma \ref{lem:criticalpts}, $(\W,(\W\W^\top)^{\frac12})\in\calE$. The converse is readily verified by substitution.
\end{proof}

It follows from LaSalle's invariance principle \citep{lasalle1968stability}, the ODE \eqref{eq:dW1} and Lemma \ref{lem:V} that when initialized  with $(\W_0,\M_0)\in\calO$, the solution $(\W(t),\M(t))$ converges to a fixed point of $V$, which corresponds to an equilibrium point in $\calE$. Next, for the general case $(\W_0,\M_0)\in\calD\setminus\calN$, we can rewrite the right-hand side of the ODE \eqref{eq:dW} as follows:
\begin{align*}
    \frac{\dd \W(t)}{\dd t}&=-2\nabla V(\W(t))+\left[\M(t)^{-1}-(\W(t)\W(t)^\top)^{-\frac12}\right]\W(t)\A,
\end{align*}
where we recall that $\W(t)\W(t)^\top$ is non-singular for all $t\ge0$ by Lemma \ref{lem:setN}. By the chain rule,
\begin{align*}%\label{eq:dV}
    \frac{\dd V(\W(t))}{\dd t}&\le-2\|\nabla V(\W(t))\|^2+\|\nabla V(\W(t))\A\W(t)^\top\|\|\M(t)^{-1}-(\W(t)\W(t)^\top)^{-\frac12}\|.
\end{align*}
We claim that
\begin{align*}
    \limsup_{t\to\infty}\|\nabla V(\W(t))\A\W(t)^\top\|\|\M(t)^{-1}-(\W(t)\W(t)^\top)^{-\frac12}\|=0.
\end{align*}
Assuming the claim holds, we have
\begin{align*}
    \limsup_{t\to\infty}\frac{\dd V(\W(t))}{\dd t}\le0.
\end{align*}
Therefore, by LaSalle's invariance principle, $\W(t)$ converges to the set of fixed points of $V$, which correspond to the set of equilibrium points $\calE$. We summarize this result in the following lemma, and provide a detailed proof in appendix \ref{apdx:gradflow}.
    
\begin{lemma}\label{lem:globalconv}
    Suppose $\A\in\calS_{++}^n$ and $(\W(t),\M(t))$ is a solution to ODE \eqref{eq:dW}--\eqref{eq:dM} with  $\tau=\frac12$ and initial condition $(\W_0,\M_0)\in\calD\setminus\calN$. Then $(\W(t),\M(t))$ converges to the set $\calE$ as $t\to\infty$. 
\end{lemma}
 
\subsection{Proof of Theorem \ref{thm:aeconvergence}}\label{sec:proof}

Existence and uniqueness of solutions is shown in Theorem \ref{thm:exist} of appendix \ref{apdx:exist}. We now combine Lemma \ref{lem:globalconv} and Theorem \ref{thm:localstability} to prove global convergence. Define the subset
\begin{align*}
    \calU:=\left\{(\W_0,\M_0)\in\calD:(\W(t),\M(t))\text{ converges to $\calE\setminus\calE_0$ as $t\to\infty$}\right\}
\end{align*}
of initializations whose trajectories converge to the set of linearly unstable equilibrium points. If $(\W_\ast,\M_\ast)\in\calE\setminus\calE_0$, then by Theorem \ref{thm:localstability}, the Jacobian of $\G$ evaluated at $(\W_\ast,\M_\ast)$ has an eigenvalue with positive real part. Therefore, by \cite[Proposition 3]{potrie2009local}, the set $\calU$ has Lebesgue measure zero. Along with Lemma \ref{lem:setN}, this implies that the set $\calZ:=\calN\cup\calU$ also has Lebesgue measure zero. Suppose $(\W(t),\M(t))$ is a solution of the ODE \eqref{eq:dW}--\eqref{eq:dM} with initial condition $(\W_0,\M_0)\in\calD\setminus\calZ$. Then by Lemma \ref{lem:globalconv} and the definition of $\calU$, $(\W(t),\M(t))$ converges to the set $\calE_0$ as $t\to\infty$. Finally, by Lemma \ref{lem:criticalpts} and the definition of $\calE_0$, for every $(\W_\ast,\M_\ast)\in\calE_0$, the row vectors of $\F_\ast:=\M_\ast^{-1}\W_\ast$ are orthonormal and span the principal subspace of $\A$.

\section{Comparing the ODE and the online algorithm}
\label{sec:stochastic}

We now use numerical simulations to examine whether the two-phase convergence observed in the ODE \eqref{eq:dW}--\eqref{eq:dM} also appears in the online algorithm \eqref{eq:DW}--\eqref{eq:DM}.

\paragraph{Setup.} We consider a network with $d=4$ inputs and $k=2$ neurons.
The inputs $\{\x_t\}$ are sampled i.i.d.\ from a mean zero normal distribution with covariance matrix $\A=\text{diag}(0.5,0.25,0.2,0.05)$.
The feedforward weight matrix $\W$ is initialized to have i.i.d.\ standard normal entries and the recurrent weight matrix $\M$ is initialized to be diagonal with entries sampled uniformly from $[1,2]$.
The ODE \eqref{eq:dW}--\eqref{eq:dM} is solved using the built-in Mathematica function \texttt{NDSolve}.
The online algorithm \eqref{eq:DW}--\eqref{eq:DM} is simulated with time-dependent step size $\eta_t=\frac{c_0}{c_1+t}$, where $c_0,c_1>0$ are chosen so that $\eta_1=0.001$ and $\sum_{t=1}^{25,000}\eta_t=8$.
To track convergence, we evaluate the Lyapunov functions $L(\W,\M)$ and $V_\ast(\W)=V(\W)-V(\W_\ast)$, where $(\W_\ast,\M_\ast)$ is any stable critical point of the ODE.
The ODE and online algorithm are simulated with 100 random initializations.

\begin{figure}
    \centering
    \includegraphics[width=1.0
    \linewidth]{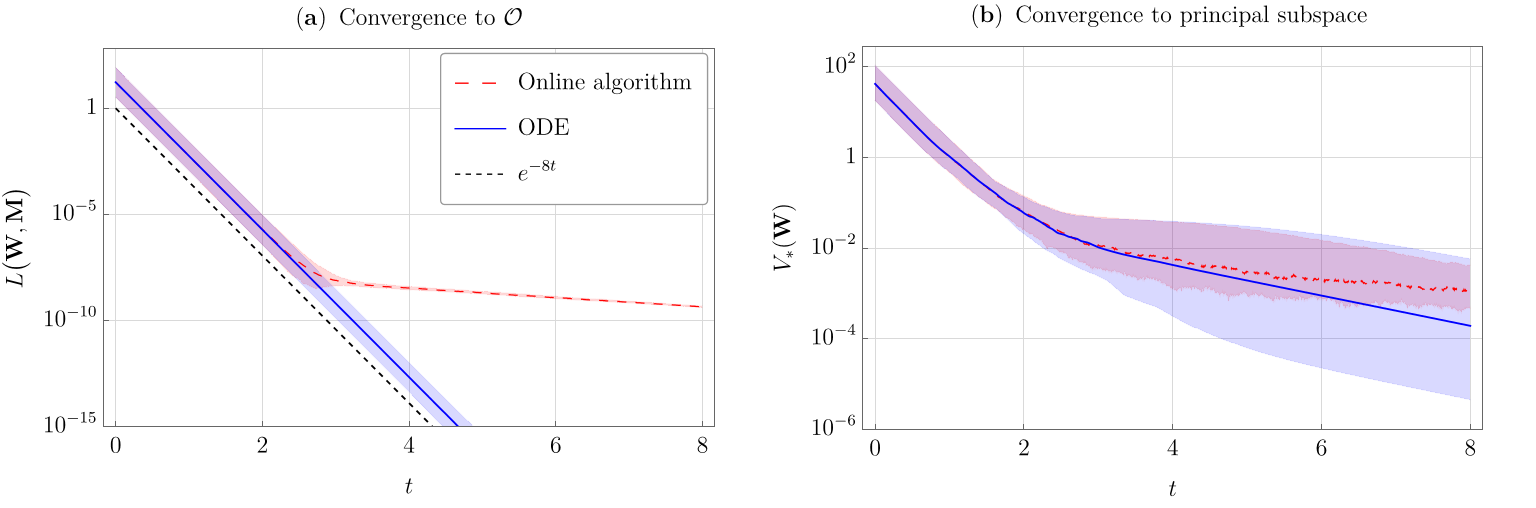}
    \caption{
    Convergence of the ODE and the online algorithm in a network with $d=4$ inputs and $k=2$ neurons. {\bf (a)} Convergence to the invariant manifold $\calO$, measured using the Lyapunov function $L(\W,\M)$. 
    {\bf (b)} Convergence to the principal subspace, measured using $V_\ast(\W)=V(\W)-V(\W_\ast)$, where $(\W_\ast,\M_\ast)$ is any stable equilibrium point of the ODE.
    In panels \textbf{(a)} and \textbf{(b)}, shaded regions indicate the middle 80th percentile over 100 random initializations; solid/dotted lines denote the median values.
    Simulation details are in the main text.
    }
    \label{fig:randomstart}
\end{figure}

\paragraph{Results.} Figure \ref{fig:randomstart} illustrates the two-phase convergence for both the ODE and the online algorithm.
Panel \textbf{(a)} shows that $L(\W,\M)$ converges exponentially with decay rate $e^{-8t}$ (dotted black line) for the ODE (solid blue line), consistent with Lemma \ref{lem:WMconv}.
For the online algorithm (dashed red line), $L(\W,\M)$ initially converges exponentially with decay rate $e^{-8t}$ before leveling off (around $t=2.5$) due to stochastic fluctuations.
Panel \textbf{(b)} shows convergence to the principal subspace via decay of $V_\ast(\W)$ for the ODE and online algorithm.
After the trajectories reach a neighborhood of $\calO$ around $t=2$, the decay rate of $V_\ast(\W)$ for both the ODE and online algorithm continues more gradually as the dynamics of $\W$ are dominated by the gradient $\nabla V(\W)$.
Overall, these simulations suggest that the online algorithm also evolves according to the same two phases as the ODE.

\section{Global stability for general synaptic learning rates}
\label{sec:tau}

Our analysis focused on the special case $\tau=\frac12$, where feedforward and recurrent synapses evolve at equal rates.
The choice simplified the structure of the invariant manifold $\calO$ and Lyapunov function $L$, enabling a transparent proof of global convergence.
However, the assumption $\tau=\frac12$ may appear restrictive and raises the question of whether the same two-phase convergence holds for other $0<\tau\le\frac12$.

\citet{centorrino2025similarity} recently established global stability in the limit $\tau\to0$, where the recurrent synapses evolve on a separate timescale than the feedforward synapses.
In this regime, the synaptic dynamics also exhibit a two-phase structure, but with separate timescales.
In the first phase, the feedforward weights $\W$ are fixed while the recurrent weights $\M$ evolve towards the invariant manifold $\calO_0=\{(\W,\M)\in\calD:\M^{-1}\W\A\W^\top\M^{-1}=\M\}$.
As shown in appendix \ref{apdx:lyapunov}, this convergence can be described by the convex Lyapunov function $L_0(\W,\M)=\|(\W\A\W^\top)^2-\M^3\|^2$.
In the second phase, as shown in \citep[section 5.3]{centorrino2025similarity}, the synaptic matrices $(\W,\M)$ evolve within the invariant manifold $\calO_0$ and the feedforward synapses $\W$ follow the gradient flow of the non-convex potential function $V_0(\W)=\frac32\|(\W\A\W^\top)^{\frac13}\|^2-\|\W\|^2$. 
Notably, both the invariant manifold $\calO_0$ and functions $L_0(\W,\M)$ and $V_0(\W)$ are distinct from those analyzed here for the case that $\tau=\frac12$.

We conjecture that global convergence holds for all $0<\tau\le\frac12$, with dynamics exhibiting a similar two-phase structure: (i) rapid convergence to a generalized invariant manifold $\calO_\tau$ with convergence described by a generalized convex Lyapunov function $L_\tau(\W,\M)$, and (ii) slower convergence along this manifold to the principal subspace that follows the gradient flow of a generalized non-convex potential function $V_\tau(\W)$.
Moreover, we conjecture that for each $0\le\tau\le\frac12$, there is a function $\Phi_\tau:\R^{k\times n}\to\calS_{+}^k$ mapping feedforward weights to recurrent weights and $p_\tau>0$ such that 
\begin{itemize}
    \item the invariant manifold satisfies $\calO_\tau=\{(\W,\M)\in\calD:\Phi_\tau(\W)=\M\}$,
    \item the convex Lyapunov function satisfies $L_\tau(\W,\M)=\|\Phi_\tau(\W)^{p_\tau}-\M^{p_\tau}\|^2$,
    \item and the potential function is $V_\tau(\W)=f(\W,\Phi_\tau(\W))$, where $f(\W,\M)$ is defined in \eqref{eq:minimax}.
\end{itemize}
Furthermore, $\Phi_\tau$ and $p_\tau$ satisfy $\Phi_{\frac12}(\W)=(\W\W^\top)^{\frac12}$ and $p_{\frac12}=2$, and $\lim_{\tau\to0}\Phi_\tau(\W)=(\W\A\W^\top)^{\frac13}$ and $\lim_{\tau\to0}p_\tau=3$.
We anticipate that the conjecture can be proved for $\tau$ in a small neighborhood of $\frac12$ using an application of the following implicit function theorem for dynamical systems (Persistence of Normally Hyperbolic Invariant Manifolds under Perturbations \citep{HirschPughShub1977}). The challenge is extending the proof beyond a neighborhood of $\tau=\frac12$.
In appendix \ref{apdx:invariant}, we provide a numerical method for estimating the invariant manifolds $\calO_\tau$ and illustrate the manifolds in the scalar setting $n=k=1$.

\section{Discussion}

We proved global convergence of solutions to the ODE \eqref{eq:dW}--\eqref{eq:dM} when $\tau=\frac12$. Our analysis revealed a two-phase structure to the convergence: (1) rapid convergence to an invariant set where the neural filters are orthonormal, and (2) slower evolution along this manifold following the gradient of a non-convex potential function whose minima correspond to the principal subspace. This result provides a rigorous link between local Hebbian/anti-Hebbian interactions and stable network-level computations. Practically, the results suggest that initializing synaptic weights on the invariant manifold (e.g., setting $\M_0=\I_k$ and $\W_0$ to have orthonormal row vectors) could accelerate convergence.

This work is an important step towards proving convergence rate guarantees for the online algorithm analogous to the results established by \citet{chou2020ode} for Oja's PCA model of a neuron. The analysis may also provide insight into the global dynamics of related algorithms that can be implemented in neural networks with local Hebbian synaptic learning rules for solving non-negative matrix factorization problem \citep{pehlevan2019neuroscience} and networks with local non-Hebbian synaptic learning rules for solving symmetric generalized eigen-subspace problems such as canonical correlation analysis \citep{lipshutz2021biologically,lipshutz2023normative}.

While the assumption $\tau=\frac12$ simplifies the analysis and enables closed-form Lyapunov functions, it is also restrictive and may not be biologically realistic.
Future work should extend the proof of global stability to the range $0<\tau\le\frac12$, as supported by numerical evidence (section \ref{sec:tau} and appendix \ref{apdx:tau}) and analytical results for the regime $\tau\to0$ \citep{centorrino2025similarity}.

\appendix
\section{Existence and uniqueness of solutions}\label{apdx:exist}

In this section, we prove existence and uniqueness of solutions to the ODE \eqref{eq:dW}--\eqref{eq:dM} for any $\tau>0$. A solution of the ODE is a continuously differentiable function $t\mapsto(\W(t),\M(t))$ from $[0,\infty)$ to $\calD$ whose derivative satisfies equations \eqref{eq:dW}--\eqref{eq:dM}. Recall that $\|\cdot\|$ denotes the Frobenius norm.

\begin{theorem}\label{thm:exist}
    Suppose $\A\in\calS_{++}^n$. For any $\tau>0$ and $(\W_0,\M_0)\in\calD$, there exists a unique solution $(\W(t),\M(t))$ to the ODE \eqref{eq:dW}--\eqref{eq:dM} with initial condition $(\W_0,\M_0)$.
\end{theorem}

\begin{proof}
For each $K<\infty$, define the set 
    $$\calD_K:=\left\{(\W,\M)\in\calD:\|\M^{-1}\|+\|\W\|<K\right\}.$$
Since $\G$ is analytic on $\calD$ and, for $K<\infty$, the closure of $\calD_K$ is compact, it follows that $\G$ is uniformly Lipschitz continuous on $\calD_K$. Let $(\W_0,\M_0)\in\calD$. Then for each $K<\infty$ sufficiently large such that $(\W_0,\M_0)\in\calD_K$, there exists a unique solution $(\W(t),\M(t))$ of the ODE \eqref{eq:dW}--\eqref{eq:dM} on the interval $[0,T_K)$, where
    $$T_K:=\inf\left\{t\ge0:\|\M(t)^{-1}\|+\|\W(t)\|\ge K\right\}$$
is the first time $(\W(t),\M(t))$ exits the set $\calD_K$. We are left to show that $T_\infty:=\lim_{K\to\infty}T_K=\infty$.

Let $\v\in\R^k$ be an arbitrary unit vector and define $a(t):=\v^\top\M(t)\v$ for all $t\in[0,T_\infty)$. Then
\begin{align*}
    \tau\frac{\dd a(t)}{\dd t}=\v^\top\M(t)^{-1}\W(t)\A\W(t)^\top\M(t)^{-1}\v-a(t)\ge-a(t),
\end{align*}
and so $a(t)\ge a(0)e^{-t/\tau}$ for all $t\in[0,T_\infty)$. Since this holds for all unit vectors $\v\in\R^k$, we have $\sigma_{\min}(\M(t))\ge me^{-t/\tau}>0$ for all $t\in[0,T_\infty)$, where $m:=\sigma_{\min}(\M_0)$ and $\sigma_{\min}(\M)>0$ denotes the smallest eigenvalue of $\M$. Therefore, 
\begin{align}\label{eq:Minvbound}
    \|\M(t)^{-1}\|=\sqrt{\tr(\M(t)^{-2})}\le \frac{\sqrt{k}}{\sigma_{\min}(\M(t))}\le \frac{\sqrt{k}}{m}e^{t/\tau},\qquad t\in[0,T_\infty).
\end{align}
Next, we have, for all $t\in[0,T_\infty)$,
\begin{align*}
    \left\|\frac{\dd\W(t)}{\dd t}\right\|&\le 2\|\M(t)^{-1}\|\|\W(t)\|\|\A\|+2\|\W(t)\|\\
    &\le 2\left(\frac{\sqrt{k}}{m}\|\A\|e^{t/\tau}+1\right)\|\W(t)\|.
\end{align*}
By Gronwall's inequality,
\begin{align}\label{eq:Wbound}
    \|\W(t)\|\le\|\W_0\|\exp\left[2\left(\frac{\sqrt{k}}{m}\|\A\|e^{t/\tau}+1\right)t\right],\qquad t\in[0,T_\infty).
\end{align}
Let $T<\infty$ be arbitrary. Setting
\begin{align*}
    K:=\max\left\{\frac{2\sqrt{k}}{m}e^{T/\tau},2\|\W_0\|\exp\left[2\left(\frac{\sqrt{k}}{m}\|\A\|e^{T/\tau}+1\right)T\right]\right\}<\infty,
\end{align*}
it follows from equations \eqref{eq:Minvbound}--\eqref{eq:Wbound} that $T_K\ge T$. Therefore, $T_\infty=\lim_{K\to\infty}T_K=\infty$.
\end{proof}

\section{Characterizing the critical points}\label{apdx:criticalpts}

In this section, we prove Lemma \ref{lem:criticalpts}, which characterizes the critical points of the ODE \eqref{eq:dW}--\eqref{eq:dM}.

\begin{proof}[Proof of Lemma \ref{lem:criticalpts}]
    First suppose $(\W_\ast,\M_\ast)\in\calE$. 
    Setting the derivative in \eqref{eq:dM} to zero, we see that $(\W_\ast,\M_\ast)$ satisfy
    \begin{align*}
        \M_\ast\inv\W_\ast\A\W_\ast^\top\M_\ast\inv&=\M_\ast.
    \end{align*}
    After left- and right-multiplying on both sides of the equality by $\M_\ast$, we obtain the relation $\M_\ast^3=\W_\ast\A\W_\ast^\top$.
    Taking cube roots yields $\M_\ast=(\W_\ast\A\W_\ast^\top)^{\frac13}$.
    Next, setting the derivative in \eqref{eq:dW} to zero, we obtain
    \begin{align*}
        \M_\ast\inv\W_\ast\A=\W_\ast.
    \end{align*}
    After right-multiplying by $\W_\ast^\top$ on both sides of the equality and substituting in with $\M_\ast\inv=(\W_\ast\A\W_\ast^\top)^{-\frac13}$, we obtain $(\W_\ast\A\W_\ast^\top)^{\frac23}=\W_\ast\W_\ast^\top$.
    Substituting in with the SVD $\W_\ast=\U\S\V^\top$ yields $\V^\top\A\V=\S$. Since $\S$ is diagonal, it follows that the column vectors of $\V$ are eigenvectors of $\A$ and the diagonal elements of $\S$ are the corresponding eigenvalues. Finally, we see that $\M_\ast=\U\S\U^\top$. On the other hand, if $\W_\ast$ and $\M_\ast$ are as in the statement of the lemma, then it is readily verified by substitution that $\G(\W_\ast,\M_\ast)={\bf 0}$.
\end{proof}

\section{Bounding the synaptic weights}\label{apdx:setN}

In this section, we prove Lemma \ref{lem:setN}, which states that the set $\calN$ of initial conditions with divergent trajectories, has Lebesgue measure zero. 
Through the remainder of the Appendices we fix $\tau=\tfrac12$.
We first need the following lemma.

\begin{lemma}\label{lem:detM}
    Suppose $\A\in\calS_{++}^n$ and $(\W(t),\M(t))$ is a solution of the ODE \eqref{eq:dW}--\eqref{eq:dM} with initial condition
    $(\W_0(t),\M_0(t))\in\calD\setminus\calN$. Then $\liminf_{t\to\infty}\det(\M(t))>0$ and $\limsup_{t\to\infty}\|\M(t)\|<\infty$. 
\end{lemma}

\begin{proof}
    For a proof by contradiction, suppose $\liminf_{t\to\infty}\det(\M(t))=0$.
    Then, along with Lemma \ref{lem:WMconv}, this implies there is a sequence $\{t_j\}$ with $t_j\to\infty$ as $j\to\infty$ such that
        $$\lim_{j\to\infty}\det(\M(t_j)^2)=\lim_{j\to\infty}\det (\W(t_j)\W(t_j)^\top)=0$$
    and for each $j=1,2,\dots$,
    \begin{align*}
        -\infty<\frac{\dd}{\dd t}\log\det(\W(t_j)\W(t_j)^\top))<0.
    \end{align*}
    Since $(\W, \M)\notin \calN$, we can choose $j$ sufficiently large such that 
    \begin{align}\label{eq:contradiction}
        \frac{k}{\sigma_{\min}(\A)}<\tr\left[\M(t_j)\inv\right]<\infty.
    \end{align}
    However, by the ODE \eqref{eq:dW}, 
    \begin{align*}
        \frac{\dd}{\dd t}\log(\det(\W(t_j)\W(t_j)^\top)) & =\tr\left[\left(\W(t_j)\W(t_j)^\top\right)\inv \frac{\dd\W(t_j)\W(t_j)^\top}{\dd t}\right] \\
        % & = 2\tr\left[\left(\W(t)\W(t)^\top\right)\inv\left(\M(t)\inv\W(t)\A\W(t)^\top-\W(t)\W(t)^\top\right)\right] \\
        & = 2\tr\left[\left(\W(t_j)\W(t_j)^\top\right)\inv\M(t_j)\inv\W(t_j)\A\W(t_j)^\top-\I_k\right] \\
        & \ge 2\tr\left[\sigma_{\min}(\A)\left(\W(t_j)\W(t_j)^\top\right)\inv\M(t_j)\inv\W(t_j)\W(t_j)^\top-\I_k\right] \\
        % & \ge 2\tr\left[\sigma_{\min}(\A)\M(t)\inv-\I_k\right] \\
        & > 0,
    \end{align*}
    which contradicts equation \eqref{eq:contradiction}.
    Therefore, $\liminf_{t\to\infty}\det(\M(t))>0$.

    We are left to show that $\limsup_{t\to\infty}\|\M(t)\|<\infty$. Again, for a proof by contradiction, suppose $\limsup_{t\to\infty}\|\M(t)\|=\infty$. 
    By the previous result, $K=\sup_{t\ge0}\|\M(t)^{-1}\|<\infty$.
    Thus, there exists $t>0$ sufficiently large such that
    \begin{align}\label{eq:dM2dt}
        \frac{\dd\|\M(t)\|^2}{\dd t}>0
    \end{align}
    and
    \begin{align*}
        -\infty<\tr\left[\sigma_{\max}(\A)\M(t) -\M(t)^2\right]+ 4KL(\W(0),\M(0)) e^{-4t}\sigma_{\max}(\A)<0.
    \end{align*}
    By the ODE \eqref{eq:dM} and Lemma \ref{lem:WMconv},
    \begin{align*}
        \frac{\dd\|\M(t)\|^2}{\dd t} &=4\tr\left[\M(t)\inv\W(t)\A\W(t)^\top-\M(t)^2\right]   \\
        &\leq 4\tr\left[\sigma_{\max}(\A)\M(t)\inv\W(t)\W(t)^\top-\M(t)^2\right]  \\
        &\leq 4\tr\left[\sigma_{\max}(\A)\M(t) -\M(t)^2\right]+ 4KL(\W(0),\M(0)) e^{-4t}\sigma_{\max}(\A)\\
        &<0.
    \end{align*}
    This contradict equation \eqref{eq:dM2dt}.
    Therefore, $\limsup_{t\to\infty}\|\M(t)\|<\infty$.
\end{proof}

\begin{proof}[Proof of Lemma \ref{lem:setN}]
    The fact that $\calN$ has Lebesgue measure zero follows immediately because $\calN\subset\{(\W,\M)\in\calD:\det(\W\W^\top)=0\}$, $\det(\cdot)$ is a nonzero polynomial, and the vanishing set of a nonzero polynomial has Lebesgue measure zero \citep{caron2005zero}. 

    Let $(\W(t),\M(t))$ be a solution of the ODE \eqref{eq:dW}--\eqref{eq:dM} with initial condition $(\W_0,\M_0)$.
    In the case $(\W_0,\M_0)\in\calD\setminus\calN$, Lemmas \ref{lem:WMconv} and \ref{lem:detM} imply that $\|\M(t)^{-1}\|$ and $\|\W(t)\|$ are uniformly bounded in $t$.
    On the other hand, suppose $(\W_0,\M_0)\in\calN$, and $(\W(t),\M(t))$ is the solution to \eqref{eq:dW}--\eqref{eq:dM} with initial condition $(\W_0,\M_0)$.  Then there exists $t>0$, $\v\in\R^k$ and $\lambda\in\R$ such that $\W(t)^\top\v=0$ and $\M(t)\v=\lambda\v$.
    Then
    \begin{align*}
        \frac{\dd \v^\top\W(t)}{\dd t}&=2\left(\v^\top\W(t)\M(t)\inv\A-\v^\top\W(t)\right)=0\\
        \frac{\dd \M(t)\v}{\dd t}&=2\left(\M(t)\inv\W(t)\A\W(t)^\top\M(t)\inv\v-\M(t)\v(t)\right)=-2\lambda\v.
    \end{align*}
    It follows that for $t>0$, $\W(t)^\top\v=0$ and $\M(t)\v=\lambda e^{-2t}\v$, $(\W(t),\M(t))\in\calN$. Therefore, for each $t>0$, $\v$ is an eigenvector of $\M(t)$ with eigenvalue $\lambda e^{-2t}$. 
    It follows that $\lim_{t\to\infty}\det(\M(t))=0$. 
\end{proof}

\section{Convergence of neural filters to an eigen-subspace}\label{apdx:gradflow}

In this section, we prove Lemma \ref{lem:globalconv}.

\begin{proof}[Proof of Lemma \ref{lem:globalconv}]
    Suppose $(\W(t),\M(t))$ is a solution to ODE \eqref{eq:dW}--\eqref{eq:dM}.
    By the definition of $V(\W)$ in equation \eqref{eq:V} and the ODE \eqref{eq:dW}, we have
    (suppressing the dependence on $t$)
    \begin{align*}
        \tr\left[\nabla_{\W}(V(\W))\frac{\dd\W^\top}{\dd t}\right]&=2\tr\left[\left(\W-(\W\W^\top)^{-\frac12}\W\A\right)\left(\A\W^\top\M\inv-\W\right)\right] \\
        &=2\tr\left[\left(\W-(\W\W^\top)^{-\frac12}\W\A\right)\left(\A\W^\top(\W\W^\top)^{-\frac12}-\W\right)\right] \\
         &\qquad+\tr\left[\left(\W-(\W\W^\top)^{-\frac12}\W\A\right)\A\W^\top\left(\M\inv-(\W\W^\top)^{-\frac12}\right)\right] \\
        &=-2\|\nabla_{\W}(V(\W^\top))\|^2 \\
         &\qquad+\tr\left[\left(\W-(\W\W^\top)^{-\frac12}\W\A\right)\A\W^\top\left(\M\inv-(\W\W^\top)^{-\frac12}\right)\right].
    \end{align*}
    By Lemma \ref{lem:setN}, $\W(t)-(\W(t)\W(t)^\top)^{-\frac12}\W(t)\A$ is uniformly bounded in $t$.
    By Lemmas \ref{lem:WMconv} and \ref{lem:setN},
    \begin{multline*}
        \limsup_{t\to\infty}\left\|\M(t)\inv-(\W(t)\W(t)^\top)^{-\frac12}\right\|^2 \\\leq\limsup_{t\to\infty}\left\|\M(t)\inv\right\|^2\left\|(\W(t)\W(t)^\top)^\frac12\right\|^2\left\|\M(t)-(\W(t)\W(t)^\top)^{\frac12}\right\|^2=0.
    \end{multline*}
    Therefore, along with LaSalle's invariance principle and Lemma \ref{lem:V}, this implies that $(\W(t),\M(t))$ converges to the set $\calE$ as $t\to\infty$.
    This completes the proof.
\end{proof} 

\section{Stability of the ODE for general timescales}
\label{apdx:tau}

In section \ref{sec:tau} we conjectured that the global convergence result proved in the main text for $\tau=\frac12$ (Theorem \ref{thm:aeconvergence}) extends to all $0<\tau\le\frac12$, and that the dynamics retain the same two-phase structure: (i) rapid convergence towards a $\tau$-dependent invariant manifold $\calO_\tau$, captured by a convex Lyapunov function $L_\tau$, and (ii) slow evolution along $\calO_\tau$ governed by the gradient flow of a $\tau$-dependent nonconvex potential $V_\tau$, whose minima correspond to the principal subspace.
Here we provide some limited empirical evidence in support of our conjecture.

\subsection{Convex Lyapunov function for the case $\tau\to0$}
\label{apdx:lyapunov}

In the $\tau\to0$ regime, the recurrent weights first converge to the invariant manifold $\calO_0=\{(\W,\M)\in\calD:\M^{-1}\W\A\W^\top\M^{-1}=\M\}$ while the feedforward weights remain fixed, after which both weights evolve within the invariant manifold.
Here we sketch out an argument showing exponential convergence in the first phase, where we assume that the feedforward weights $\W$ remain fixed; that is, $\frac{\dd\W}{\dd t}=0$.
Recall the convex Lyapunov function $L_0(\W,\M)=\|\W\A\W^\top-\M^3\|^2$ from section \ref{sec:tau}.
Suppose $\W$ is fixed and $\M(t)$ is a solution of the ODE \eqref{eq:dM}.
By the product rule and the ODE \eqref{eq:dM}, and suppressing the dependence on $t$, we have
\begin{align*}
    \tau\frac{\dd}{\dd t}\M^3&=\M\W\A\W^\top\M^{-1}+\W\A\W^\top+\M^{-1}\W\A\W^\top\M-3\M^3.
\end{align*}
Next, by the chain rule, the previous display and the cyclic property of the trace operator, we have
\begin{align*}
    \tau\frac{\dd }{\dd t}L_0(\W,\M)
    &=-2\tr\left[\left(\W\A\W^\top-\M^3\right)\left(2\M\W\A\W^\top\M^{-1}+\W\A\W^\top-3\M^3\right)\right]\\
    &=-4\left\|\M^{\frac12}\left(\W\A\W^\top-\M^3\right)\M^{-\frac12}\right\|^2-2L_0(\W,\M)\\
    &\le-2L_0(\W,\M).
\end{align*}
Therefore, in the first phase, the Lyapunov function $L_0(\W,\M)$ converges exponentially to zero.

\subsection{Estimation of invariant manifolds in the scalar setting}
\label{apdx:invariant}

\begin{figure}
     \centering
     \includegraphics[width=1.0\linewidth]{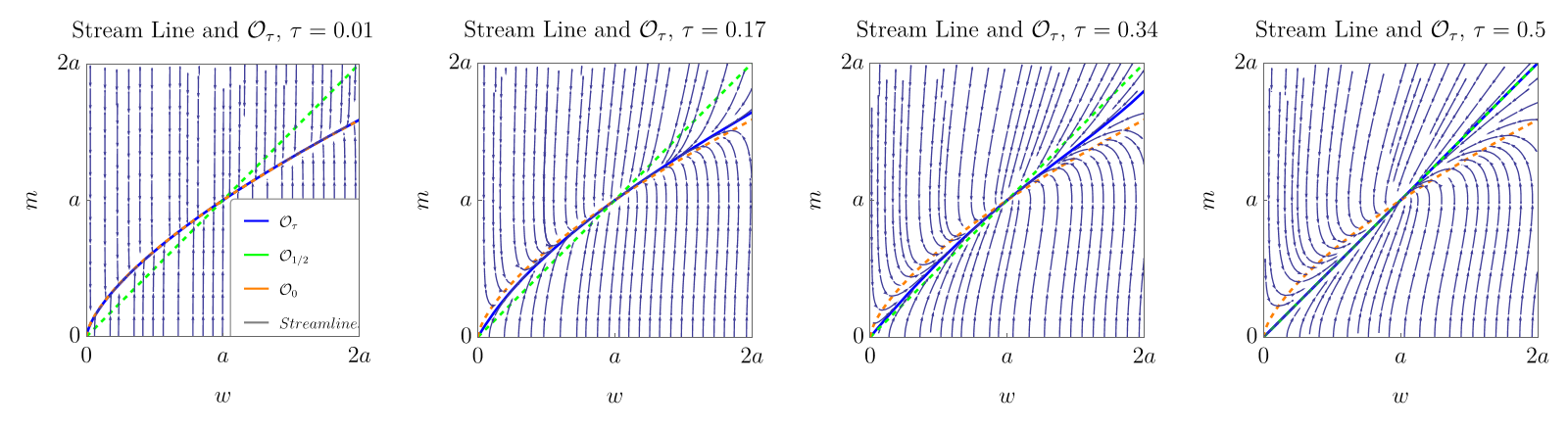}
     \caption{For each $\tau$, the arrows show the streamlines for the ODE \eqref{eq:dW}--\eqref{eq:dM}. The solid blue line is a numerical estimate of the invariant set $\calO_\tau$. The orange and green dashed lines are the invariant sets $\calO_0$ and $\calO_\frac12$, respectively.}
     \label{fig:dk1OtauStreamLine}
\end{figure}

Consider the scalar setting $n=k=1$, in which case the critical points are $\pm(a,a)$.
Figure \ref{fig:dk1OtauStreamLine} shows the streamlines of the ODE \eqref{eq:dW}--\eqref{eq:dM} for $w>0$, computed using Mathematica's built-in function \texttt{Streamline} function, converging for different $\tau\in(0,\frac12]$. 
The invariant manifolds $\calO_0$ and $\calO_\frac12$ are shown in orange and green, respectively. For intermediate $\tau$, we numerically estimate $\calO_\tau$ using the method described in the next paragraph. Note that $\calO_\tau$ appears to interpolate between $\calO_0$ and $\calO_\frac12$ as $\tau$ varies from 0 to $\frac12$, and this is supported by further simulations (not shown).

\paragraph{Numerical estimation of $\calO_\tau$ for $w>0$.} To estimate $\calO_\tau$ when $a<w<2a$, we define a function $h:(0,\infty)^2\to\R^2$ as follows. Given a point $(w_0,m_0)$, let $(w(t),m(t))$ denote the solution of the ODE with initial condition $(w_0,m_0)$ and set $h(w_0,m_0)=\lim_{t\to\infty}(\dot{w}(t),\dot{m}(t))$. A point $(w_0,m_0)$ is in $\calO_\tau$ if $h(w_0,m_0)$ is parallel to the eigenvector of the Jacobian of the vector field $\G$ with smallest associated eigenvalue. We then perform a numerical search to find a point $(w_0,m_0)\in\calO_\tau$ on the line $w=2a$, and then compute the solution of the ODE starting at $(w_0,m_0)$ to estimate $\calO_\tau$ for $a<w<2a$. We implemented this procedure with the help of Mathematica's built-in functions \texttt{NDSolve} and \texttt{FindRoot}. To estimate $\calO_\tau$ for $0<w<a$, we could use a similar argument starting with points $(w_0,m_0)$ near the origin.
However, there is an alternative approach that does not require a numerical search. 
Rather, we assume that the origin is in the closure of the invariant manifold $\calO_\tau$ and the invariant manifold can be linearly approximated when $m$ is small; that is the invariant manifold near the origin is well-approximated by $\{(v_0m,m):m\approx 0\}$ for some $v_0>0$. Under this assumption, when $(w_0,m_0)$ are small, the derivatives of the solutions to the ODE \eqref{eq:dW}--\eqref{eq:dM} at $t=0$ are approximately $\dot w(0)\approx \frac{2aw_0}{m_0}$ and $\dot m(0)\approx \frac{aw_0^2}{\tau m_0^2}$. Using the fact that $w_0\approx v_0m_0$ and $\dot w(0)\approx v_0\dot m(0)$, we see that $2av_0\approx\frac{a}{\tau}v_0^3$ and so $v_0=\sqrt{2\tau}$. Therefore, we estimate $\calO_\tau$ by computing a solution starting at $(\sqrt{2\tau}\alpha,\alpha)$ for $\alpha>0$ small.
This numerical procedure can also be generalized to estimate the invariant manifolds when $d,k>1$; however, visualizing the manifolds is more difficult.

\bibliographystyle{unsrtnat}
\bibliography{stable}

\end{document}